\documentclass[letterpaper,onecolumn,accepted=2025-10-22]{quantumarticle}
\pdfoutput=1

\usepackage{amsmath}

\usepackage{amssymb}
\usepackage{amsthm}
\usepackage{color}
\usepackage{comment}

\usepackage{hyperref}

\usepackage{tikz}
\usetikzlibrary{quantikz}

\renewcommand{\>}{\rangle}
\newcommand{\<}{\langle}

\newcommand{\cN}{\mathcal{N}}

\newcommand{\C}{\mathbb{C}}
\newcommand{\R}{\mathbb{R}}

\newcommand{\tr}{\mathrm{tr}}

\newcommand{\cC}{\mathcal{C}}

\newcommand{\cT}{\mathcal{T}}

\newcommand{\cR}{\mathcal{R}}

\newcommand{\cX}{\mathcal{X}}

\newtheorem{theorem}{Theorem}
\newtheorem{definition}[theorem]{Definition}

\newtheorem{proposition}[theorem]{Proposition}
\newtheorem{lemma}[theorem]{Lemma}
\newtheorem{remark}[theorem]{Remark}

\begin{document}

\title{Quantized Markov Chain Couplings that Prepare Qsamples}

\author{Kristan Temme}
\affiliation{IBM Quantum, IBM Thomas J. Watson Research Center, Yorktown Heights, NY, USA}

\author{Pawel Wocjan}
\affiliation{IBM Quantum, IBM Thomas J. Watson Research Center, Yorktown Heights, NY, USA}

\maketitle

\abstract{We present a novel approach to quantizing Markov chains. The approach is based on the Markov chain coupling method, which is frequently used to prove fast mixing. Given a particular coupling, e.g., a grand coupling, we construct a completely positive and trace preserving map. This quantum map has a unique fixed point, which corresponds to the quantum sample (qsample) of the classical Markov chain's stationary distribution. We show that the convergence time of the quantum map is directly related to the coupling time of the Markov chain coupling.}


\section{Introduction}
 It is known that ergodic (that is, irreducible and aperiodic) Markov chains on finite state spaces converge to their stationary distributions. Let $P=(p_{x',x})$ denote the transition matrix of such a Markov chain with state space $\Omega$. We choose the convention that the stochastic matrix $P$ acts on probability distributions written as column vectors. Thus, $P$ is column stochastic. Let $\pi = (\pi_x)$ denote the stationary distribution.

In the classical setting, we have
\begin{align}
    \lim_{m\rightarrow\infty} P^m q 
    &= 
    \pi,
\end{align}
where $q$ is an arbitrary initial probability distribution. The mixing time $t_\mathrm{mix}$ is the smallest $m$ so that $P^m q$ is sufficiently close to the stationary distribution $\pi$ (say $1/4$-close in the total variation distance\footnote{Once the total variation distance falls below $1/4$, it can be made $\varepsilon$ small by running the Markov chains for $O\big(\log(\varepsilon^{-1}) \,t_\mathrm{mix}\big)$ many steps.}). The Markov chain is called fast mixing if the mixing time scales polylogarithmically in the size $N=|\Omega|$ of the state space $\Omega$.

There are different methods for quantizing classical Markov chains to create quantum analogs. A common approach is constructing a unitary referred to as a canonical quantum walk \cite{childs2004spatial}, where the transition matrix is converted into a Hamiltonian evolution. Another common approach is constructing a unitary referred to as a Szegedy walk \cite{szegedy2004quantum}, which is the product of two reflections defined in terms of the transition matrix. Here, we take the approach that instead of a using a classical transition matrix, quantum Markov chains are modeled by completely positive and trace preserving maps (tcp-maps) \cite{wolf2011url}, which are linear maps $\cT$ acting on quantum state spaces. Starting from any initial density matrix, these maps converge to their fixed point when the relevant conditions are met. It has been shown that such processes are universal for quantum computing \cite{verstraete2009quantum}. They can be used to prepare physically interesting pure states \cite{diehl2008quantum}, as well as thermal states \cite{temme2011quantum,rall2023thermal,chen2023quantum,jiang2024quantum}. \\

{\it First}, starting from a classical Markov chain $P$ with fixed point $\pi$, we construct a family of completely positive and trace preserving maps (tcp-maps) $\cT$ that converge to the qsample $|\sqrt{\pi}\>$, which is defined by
\begin{align}
    |\sqrt{\pi}\> 
    &= 
    \sum_{x \in \Omega} \sqrt{\pi_x} |x\>.
\end{align}
Specifically, we construct a tcp-map $\cT$ based on a so called coupling \cite{levin2017mixing} for the Markov chain so that 
\begin{align}
    \lim_{m\rightarrow\infty}\cT^m(\rho) &= |\sqrt{\pi}\>\<\sqrt{\pi}|
\end{align}
for any initial quantum state $\rho$. We ask that the qsample $|\sqrt{\pi}\>$ is the unique fixed point of $\cT$. Qsamples are important ingredients for several quantum algorithms \cite{aharonov2003adiabatic}, and this procedure can be used as a building block. {\it Second}, we bound the convergence to the fixed point. We bound how fast the trace distance, or Schatten 1-norm, $\|A\|_\mathrm{tr} = \sum_i \sigma_i(A)$ given as the sum of the singular values of $A$, decays for $\| \cT^{m}(\rho) - |\sqrt{\pi}\>\<\sqrt{\pi}|\|_\mathrm{tr}$.\\

Central to our construction is that many important classical Markov chains in the literature are shown to be fast mixing by constructing suitable Markov chain couplings \cite{levin2017mixing}. Couplings also play an important role in the construction of algorithms, such as the coupling from the past algorithm by Propp and Wilson \cite{propp1996exact}, where they are used to ensure that one is sampling from the exact steady state distribution. Briefly speaking, a coupling is a Markov chain on the state space $\Omega\times\Omega$ with transition matrix $C=(c_{(x',y'),(x,y)})$ such that (i) the transitions of each component viewed separately behave according to the original transition matrix $P$ and (ii) once the two components coalesce (i.e., their states become equal), they stay together forever.  The coalescence time $\tau_\mathrm{coal}$ is the random variable that corresponds to the time it takes for the two components to coalesce. We refer to the shortest time $t_\mathrm{couple}$ as the number of steps the coupled chain takes to ensure that the probability of the two components not having coalesced is below $\Pr_{\max}(\tau_\mathrm{coal} > t_\mathrm{couple}) \leq 1/4$. This coupling time  $t_\mathrm{couple}$ is conventionally used to bound the mixing time of the classical chain $P$. 

We use such couplings to construct our novel tcp-maps from classical Markov chains and prove that the convergence time to reach the qsample fixed point is closely related to the coupling time of the classical chain.  In particular we show that

\begin{theorem}[main results (informal)]
Assume we are given a Markov chain coupling for an ergodic Markov chain $P$ with stationary distribution $\pi=(\pi_x : x \in \Omega)$.  We construct a corresponding trace preserving map $\cT$ called quantized coupling. \\ 

For independent and grand couplings, the resulting quantized coupling $\cT$ is a valid tcp-map that prepares the qsample $|\sqrt{\pi}\>\<\sqrt{\pi}|$. Its convergence is bounded by 

\begin{align}
    \frac{1}{2} \big\| \cT^{m}(\rho) - |\sqrt{\pi}\>\<\sqrt{\pi} \big\|_\mathrm{tr}
    &\le \sqrt{\varepsilon},
\end{align}
whenever 

\begin{align}
    m \geq \frac{1}{2} \log_2 \left (\frac{1}{\varepsilon \; \pi_*}\right) \; t_\mathrm{couple},
\end{align}
where $\rho$ is any initial state, $\pi_* = \min\{ \pi_x : x \in \Omega\}$.

The quantized coupling $\cT$ can be implemented efficiently for sparse grand couplings. 
\end{theorem}
\medskip

Note that this result is different from Szegedy walks \cite{szegedy2004quantum}, which are unitaries that permit the reflection about the qsample when combined with phase estimation. The ability to reflect about a state does not necessarily imply that the state can be prepared efficiently, cf.~e.g., \cite{bennett1997strengths}. The result here appears to be close to the statement that we can prepare qsamples if we have rapidly mixing classical Markov chains that sample from the distribution. However, we want to point out that the fact that a Markov chain is rapidly mixing does not immediately imply that a corresponding coupling coalesces quickly. Indeed, the former statement seems too powerful for quantum algorithms as such a result would imply that a quantum computer could efficiently solve complete problems for statistical zero knowledge proofs \cite{aharonov2003adiabatic}. Relating the preparation of qsamples to coupling times is therefore a meaningful relaxation. 

We now briefly highlight some of the key areas where qsamples are used. They play a central role in algorithms for estimating partition functions and Bayesian inference \cite{arunachalam2022simpler, harrow2023adaptive}. Qsamples are also essential in algorithms for estimating volumes of convex bodies \cite{chakrabarti2023}, sampling log-concave and non-log-concave distributions \cite{childs2022,ozgul2024}, and operator-level quantum acceleration \cite{leng2025}.

In the algorithms for partition function estimation,  Bayesian inference, and volume estimation, the following situation arises: a sequence of slowly varying Markov chains is used, where the stationary distributions of adjacent chains are close in total variation distance. This implies that the corresponding qsamples have large overlap, and the qsample for the first Markov chain in the sequence can be efficiently prepared. Both methods \cite{aharonov2003adiabatic,wocjan2008speedup} require traversing intermediate steps, akin to climbing a ladder one rung at a time. One approach constructs a sequence of Hamiltonians whose ground states are the qsamples, transforming adiabatically the initial qsample to the final one through these intermediate states \cite{aharonov2003adiabatic}. The other uses Szegedy walks and fixed-point search to achieve the same transformation step-by-step \cite{wocjan2008speedup}.

The algorithms for sampling log-concave, non-log-concave distributions, and operator-level quantum acceleration assume the existence of a ``warm start'' state with sufficiently large overlap with the desired qsample.

Our new method differs by preparing qsamples directly for a single Markov chain starting in an arbitrary initial state. While our focus in this work is on introducing this new mechanism---rather than optimizing all parameters or deriving the tightest running times---we emphasize that improving its efficiency remains an important open question. Such advancements could enable broader applications in these and other emerging quantum algorithms.

\medskip
\noindent
The paper is organized as follows.
In Section~\ref{sec:Markov_chain_couplings}, we introduce Markov chain couplings and describe how they are used to bound the mixing time. We present two types of couplings: (a) independent couplings and (b) grand couplings. 
In Section~\ref{sec:quantion_couplings}, we present our quantization method that transforms Markov chain couplings into linear maps that act on quantum states. We prove that our method converts independent and grand couplings into valid tcp-maps. In the appendix, we also present a counterexample of a certain coupling that does lead to a valid tcp-map (because the resulting linear map is not completely positive).
In Section~\ref{sec:convergence}, we analyze the convergence rate of quantized couplings, that is, how fast they prepare the qsamples. We establish a close connection to the coalescence time of the Markov chain coupling. 
In Section~\ref{sec:implementation}, we describe how the quantized couplings can be efficiently implemented. We discuss the implementation for three concrete examples: (i) random walk on hypercube,  (ii) random colorings, and (iii) hardcore model with fugacity. These three examples are all grand couplings. We also formulate a general theorem that shows when quantized couplings can be efficiently implemented for sparse grand couplings. The latter include many important examples in the literature.

\section{Markov chain couplings}\label{sec:Markov_chain_couplings}

Markov chain coupling is a powerful tool for bounding the mixing time. We give only a brief overview and refer the reader to \cite{levin2017mixing} for an in-depth introduction.

Recall that the column stochastic matrix $P=(p_{x',x})$ denotes the transition matrix of an ergodic Markov chain on the finite state space $\Omega$ with limiting distribution $\pi$. Define 
\begin{align}
    d(m) 
    &=
    \max_{x\in\Omega} \frac{1}{2} \sum_{x'\in\Omega} 
    \big| [P^m]_{x',x} - \pi_{x'} \big|,
\end{align}
where $([P^m]_{x',x} : x' \in\Omega)$ is the probability distribution that we obtain by starting the Markov chain in the initial state $x$ and making $m$ transitions. The total variation distance $d(m)$ measures how close we are to the desired stationary distribution $\pi$ as the number of steps $m$ increases. 
The mixing time $t_\mathrm{mix}(\varepsilon)$ is defined by
\begin{align}
    t_\mathrm{mix}(\varepsilon) 
    &= 
    \min\{m : d(m) \le \varepsilon \}
\end{align}
for $\varepsilon\in (0,1)$. 
It satisfies the bound
\begin{align}
    t_\mathrm{mix}(\varepsilon) 
    &\le \big\lceil \log_2\varepsilon^{-1} \big\rceil t_\mathrm{mix},
\end{align}
where $t_\mathrm{mix}=t_\mathrm{mix}(1/4)$.

A coupling is a Markov chain on state space $\Omega\times\Omega$ with transition matrix $C =(c_{(x',y'), (x,y)})$ so that the following \emph{coupling conditions} are satisfied:
\begin{enumerate}
    \item Each component viewed separately undergoes transitions according to the original transition matrix $P$:
    \begin{align}
        \sum_{y'} c_{(x',y'),(x,y)} 
        &= 
        p_{x',x} \label{eq:coupling_x} \\
        \sum_{x'} c_{(x',y'),(x,y)} 
        &= 
        p_{y',y} \label{eq:coupling_y}
    \end{align}
    \item The two components remain together forever once they have met, that is, $x'=y'$ if $x=y$:
    \begin{align}
        c_{(x',x'),(x,x)}
        &=
        p_{x'x} \\
        c_{(x',y'),(x,x)} 
        &=
        0 \mbox{ for } x'\neq y'
    \end{align}
    \item The two components satisfy the symmetry condition:
    \begin{align}
        c_{(x',y'),(x,y)} =
        c_{(y',x'),(y,x)}
    \end{align}
\end{enumerate}

\medskip
Let $(X_m,Y_m)$ denote the random state of the coupling that starts in the initial state $(x,y)$ and makes $m$ transitions according to the transition matrix $C$. Let $\tau_\mathrm{coal}$ denote the random variable  
\begin{align}
    \tau_\mathrm{coal} 
    &=
    \min\{ m : X_\ell = Y_\ell \mbox{ for all } \ell \ge m\},
\end{align}
which is called the coalescence time for the initial state $(x,y)$. To keep the notation simple, we omit $(x,y)$ when writing $\tau_\mathrm{coal}$.

We need to make an important remark at this point. Observe that the transition matrix $C$ could depend on the initial state $(x,y)$. It could be constructed in such a way that the corresponding coalescence time is reduced for this particular initial state. However, we will not consider such specialized transition matrices. We assume that the transition matrix $C$ is the same when defining the coalescence time for all initial states $(x,y)$.\\

The most important result of Markov chain coupling is summarized in the following theorem \cite{levin2017mixing}:

\begin{theorem}[Connection between coalescence and mixing times]
The total variation distance $d(m)$ and the mixing time $t_\mathrm{mix}$ are bounded from above by
\begin{align}
    d(m) 
    &\le 
    \max_{x,y} \Pr_{x,y}\big(\tau_\mathrm{coal} > m \big) \label{eq:upper_dm}\\
    t_\mathrm{mix} 
    &\le 
    4 \max_{x,y} \mathbb{E}_{x,y}(\tau_\mathrm{coal}),
\end{align}
respectively, where
\begin{itemize}
    \item $\Pr_{x,y}(\tau_\mathrm{coal} > m) $ denotes the probability that the random variable $\tau_\mathrm{coal}$ is greater than $m$ and
    \item 
    $\mathbb{E}_{x,y}(\tau_\mathrm{coal})$ denotes the expected value of the random variable $\tau_\mathrm{coal}$ 
\end{itemize}
for the initial state $(x,y)$.
\end{theorem}

Thus, the problem of showing that a particular Markov chain with transition matrix $P$ mixes fast reduces to the problem of constructing a suitable coupling with transition matrix $C$ that coalesces quickly when starting in any initial state $(x,y)$.\\

In an analogy to the mixing time $t_\mathrm{mix}$, it is convenient to introduce the coupling time $t_\mathrm{couple}$ for our analysis. Given the probability $\Pr_{x,y}(\tau_\mathrm{coal} > m) = \Pr_{x,y}(X_m \neq Y_m)$, we define the coupling time as  
\begin{align}
    t_\mathrm{couple} 
    &=
    \min\left \{ m : \max_{x,y} \Pr_{x,y}(\tau_\mathrm{coal} > m) \leq \frac{1}{4} \right\}.
\end{align}
Using the upper bound on the total variation distance $d(m)$ in (\ref{eq:upper_dm}) and the definitions of $t_\mathrm{mix}$ and $t_\mathrm{couple}$, it follows directly that $t_\mathrm{mix}\le t_\mathrm{couple}$.

Furthermore, we can relate the probability that we have not coalesced after $m$ to the transition matrix as follows:

\begin{remark}\label{rem:helpful}
For the analysis of our quantized coupling, it is very helpful to note that $\Pr_{x,y}(\tau_\mathrm{coal})$ can be equivalently written as
\begin{align}
    \Pr_{x,y}(\tau_\mathrm{coal} > m)
    &=
    \sum_{x'\neq y'} [C^m]_{(x',y'),(x,y)},
\end{align}
where $([C^m]_{(x',y'),(x,y)} : (x',y')\in \Omega\times\Omega)$ is the probability distribution that we obtain by starting the coupling in the state $(x,y)$ and running it for $m$ times steps.
\end{remark}

\subsection{Examples of couplings}\label{sec:coupling_examples}

We present independent and grand couplings.

\begin{definition}
[Independent coupling]\label{rem:indepedent_C}
The coupling $C_\mathrm{ind}=(c_{(x',y'),(x,y)})$ is defined by 
\begin{align}
    c_{(x',y'),(x,y)} 
    &=
    \left\{
        \begin{array}{cl}
            p_{x'x} \cdot p_{y'y} & \mbox{ if } x\neq y \\
            p_{x'x}               & \mbox{ if } x' = y' \mbox{ and } x = y \\
            0                     & \mbox{ if } x'\neq y' \mbox { and } x = y 
        \end{array}
    \right.
\end{align}
\end{definition}

In some sense, an independent coupling is the most basic coupling possible. If the two components are different, apply the Markov chain to both components independently. If they are equal, apply the Markov chain only to the first component and copy the resulting successor state to the second component. Independent couplings are not very useful for proving fast mixing.

We need to introduce some terminology before we can define grand couplings, which are very useful for establishing fast mixing.\\

A \emph{random mapping representation} of the transition matrix $P=(p_{x',x})$ on state space $\Omega$ is a function $f : \Omega \times \cR \rightarrow \Omega$, along with a $\cR$-valued random variable $R$, satisfying 
\begin{align}
    \Pr\big(f(x, R) = x' \big) 
    &= 
    p_{x',x}
\end{align}
for all possible pairs of current state $x$ and successor state $x'$.

It is known that every transition matrix $P$ on a finite state space has a random mapping representation. This is proved in \cite{levin2017mixing} with the following choices: the set $\cR$ is equal to the interval $[0,1]$ and the random variable $R$ is uniformly chosen in $\cR$. For our purposes, we assume that $\cR$ is a finite set; this occurs naturally for many important Markov chains.   

For $r\in\cR$, we write $\Pr(r)$ to denote the probability that the random variable $R$ takes on the value $r$. The transition matrix $P$ can be expressed as
\begin{align}
    P 
    &= 
    \sum_{r\in \cR} 
    \sum_{x\in\Omega} \Pr(r) \, 
    |f(x,r)\>\<x|.
\end{align}

\begin{definition}[Grand coupling]\label{rem:grand_C}
We use the same randomness in both components to define
the transition matrix $C_\mathrm{grand}=(c_{(x',y'),(x,y)})$ of a grand coupling. 
We set the transition probabilities to
\begin{align}
    c_{(x',y'),(x,y)}
    &=
    \Pr(f(x,R)=x' \mbox{ and } f(y, R)=y'). 
\end{align}
\end{definition}

\noindent
Observe that the transition matrix $C_\mathrm{grand}$ can be expressed as 
\begin{align}\label{eq:C_grand}
    C_\mathrm{grand}
    &=
    \sum_{r\in \cR} \sum_{(x,y)\in\Omega\times\Omega} \Pr(r) \, |f(x,r),f(y,r)\>\<x,y|.
\end{align}
It plays a central role in our quantization method, which is introduced in the next section.

\section{Quantized Markov chain couplings}\label{sec:quantion_couplings}

We define a linear map in terms of the transition matrix $C$ of an \emph{arbitrary} coupling. We prove that this linear map leads to valid tcp-maps for independent and grand couplings.

\begin{definition}[Quantized coupling]\label{def:cqc}
Let $C=(c_{(x',y'),(x,y)})$ be the transition matrix of an arbitrary coupling of an ergodic Markov chain with transition matrix $P=(p_{x',x})$ and stationary distribution $\pi=(\pi_x)$. 

Define 
\begin{align}
    \cT^*(\bullet)  
    &=
    D^{-1/2} \, \cC^*( D^{1/2} \bullet D^{1/2} ) \, D^{-1/2}
    \label{eq:T_star}
\end{align}
where 
\begin{align}
    \cC^*(\bullet) 
    &= 
    \sum_{x, y, x', y' \in\Omega} c_{(x',y'),(x,y)} |x'\>\<x| \bullet |y\>\<y'|
\end{align}
and $D=\mathrm{diag}(\pi_x : x \in \Omega)$. 

The quantized coupling $\cT$ is defined as the dual of the linear map $\cT^*$ with respect to the Hilbert-Schmidt inner product.
\end{definition}

Later, it will become apparent why we define the quantized coupling $\cT$ with the help of the dual map $\cT^*$.  The latter appears naturally in the proof establishing that the qsample $|\sqrt{\pi}\>$ is the unique fixed point and also in the derivation of the bound on the convergence rate, which relates it to the coalescence time of the coupling.

Let us now highlight some key properties of quantized couplings. First, observe that the action of the linear map $\cC^*$ in the vectorized representation is described by the transition matrix $C$
\begin{align}
    C
    &=
    \sum_{x, y, x', y' \in\Omega} c_{(x',y'),(x,y)} |x'\>\<x| \otimes |y'\>\<y| \\
    &=
    \sum_{x, y, x', y' \in\Omega} c_{(x',y'),(x,y)} |x' y'\>\<x y|,
\end{align}
where we abbreviate $|x'y'\>\equiv |x'\>\otimes |y'\>$ and $|xy\>\equiv |x\>\otimes |y\>$. This is seen using a fundamental property of the vectorization map and the property that $C$ is symmetric under exchanging both components (which is the third coupling condition).
Let $M$ be any matrix acting on $\Omega$. We use the convention that $\mathrm{vec}(M)$ is the column vector obtained by stacking the column vectors of $M$. Matrix multiplication under the vectorization behaves according to formula
\begin{align}
    \mathrm{vec}(A M B) &= (B^T \otimes A) \mathrm{vec}(M),
\end{align}
which holds for any triplet of square matrices $A$, $M$, and $B$ acting on $\Omega$. 

Second, observe that the application of the linear map $D^{1/2}\bullet D^{1/2}$ inside $\cC^*$ and the linear map $D^{-1/2}\bullet D^{-1/2}$ outside $\cC^*$ in (\ref{eq:T_star}) corresponds to the similarity transformation 
\begin{align}
    (D^{-1/2} \otimes D^{-1/2}) \, 
    C \, 
    (D^{1/2} \otimes D^{1/2})
\end{align}
in the vectorized representation.

Finally, note that it is not necessary for the mathematical definition of a quantized Markov chain coupling that the underlying chain is detailed-balanced (also called reversible). That said, detailed-balanced maps have the convenient property that the ratios of the probabilities in the stationary distribution can be expressed in terms of the transition probabilities. This can be used to simplify the algorithmic implementation.

\begin{lemma}[Trace preserving]\label{lem:trace_preserving}
The map $\cT^*$ satisfies
\begin{align}
    \cT^*(I) &= I,
\end{align}
that is, $\cT$ is trace preserving.
\end{lemma}

\begin{proof}
We have
\begin{align}
    \cT^*(I) 
    &=
    D^{-1/2} \, \cC^*( D^{1/2} I D^{1/2} ) \, D^{-1/2} \\
    &=
    D^{-1/2} \, \cC^*( D ) \, D^{-1/2} \\
    &=
    D^{-1/2} \, D \, D^{-1/2} = I.
\end{align}
In the above derivation, we used the identity $\cC^*(D)=D$. This identity follows from the definition $D=\mathrm{diag}(\pi_x : x \in \Omega) = \sum_x \pi_x |x\>\<x|$ and the coupling properties. We have
\begin{align}
    \cC^*(D) 
    &=
    \sum_{x, y, x', y' \in\Omega} 
    c_{(x',y'),(x,y)} \, |x'\>\<x| D |y\>\<y'| \\
    &=
    \sum_{x, x', y' \in\Omega} 
    c_{(x',y'),(x,x)} \, \pi_x \, |x'\>\<y'| \\
    &=
    \sum_{x, x' \in\Omega} 
    c_{(x',x'),(x,x)} \, \pi_x \, |x'\>\<x'| \\
    &=
    \sum_{x'\in\Omega} \left( 
    \sum_{x\in\Omega} 
    p_{x',x} \, \pi_x \, 
    \right) \, |x'\>\<x'| \\
    &=
    \sum_{x'\in\Omega} 
    \pi_{x'} \, 
    |x'\>\<x'| = D.
\end{align}
The proof shows that the similarity transformation converts the identity $\cC^*(D)=D$ into the identity $\cT^*(I)=I$.
\end{proof}

\begin{lemma}[Fixed point]\label{lem:fixed_point}
The qsample $\ket{\sqrt{\pi}}$ is a fixed point of the map $\cT$ so that 
\begin{align}
    \cT(\proj{\sqrt{\pi}}) &= \proj{\sqrt{\pi}}.
\end{align}
\end{lemma}

\begin{proof}
We have $\ket{\sqrt{\pi}} = D^{1/2}\ket{e}$ with $\ket{e} = (1, 1,\ldots, 1)^T = \sum_{x \in \Omega} \ket{x}$. Hence, we can write
\begin{align}
    \cT(\proj{\sqrt{\pi}}) 
    &=
    D^{1/2} \, \cC( D^{-1/2} |\sqrt{\pi}\>\<\sqrt{\pi}| D^{-1/2} ) \, D^{1/2} \\
    &=
    D^{1/2} \, \cC( \proj{e} ) \, D^{1/2}.
\end{align}
Note that we have $D^{-1/2}\bullet D^{-1/2}$ inside $\cC$ and $D^{1/2}\bullet D^{1/2}$ outside $\cC$ because we consider $\cT$ and instead of the dual map $\cT^*$.

We now use $\sum_{x',y' \in \Omega} c_{(x',y'),(x,y)} = 1$ for all $x,y\in\Omega$ since the transition matrix $C$ is column stochastic. Thus, we have
\begin{align}
\cC( \proj{e} ) 
&= 
\sum_{x,y,x',y'}
c_{(x',y'),(x,y)} |y\>\<y'|e\>\<e|x'\>\<x| \\
&= 
\sum_{x,y} |y\>\<x| \\
&=
\proj{e}.
\end{align}
Again, note that we consider $\cC$ instead of $\cC^*$, which explains why the outer products $|x'\>\<x|$ and $|y\>\<y'|$ appear on different sides.

Combining everything together, we obtain  
\begin{align}
    \cT(\proj{\sqrt{\pi}}) 
    &=
    D^{1/2} \, \cC( \proj{e} ) \, D^{1/2}. \\
     &=
    D^{1/2} \, \proj{e} \, D^{1/2} = \proj{\sqrt{\pi}}
\end{align}
establishing that $\ket{\sqrt{\pi}}$ is a fixed point of $\cT$.
\end{proof}

\subsection{Completely positive quantized couplings} 

We showed that the map $\cT$ is trace-preserving and has the qsample $|\sqrt{\pi}\>$ as a fixed point. To establish that this map is a physical quantum operation, we need to prove that it is also completely positive. Unfortunately, this is not true for all couplings $C$. We provide a counterexample (c.f.~Appendix), which shows that there exist classical couplings whose corresponding quantized couplings are not completely positive. However, product couplings and grand couplings (c.f.~Subsection~\ref{sec:coupling_examples}) always give rise to completely positive maps.

\begin{lemma}[Quantized independent coupling is completely positive]
The map $\cT$ corresponding to an independent coupling is completely positive.
\end{lemma}

\begin{proof}
It suffices to check that the map $\cC^*$ is completely positive since:
\begin{itemize}
    \item[(i)] $\cT$ is completely positive if and only its dual $\cT^*$ is completely positive
    \item[(ii)] the similarity transformation between $\cC^*$ to $\cT^*$ preserves complete positiveness.
\end{itemize}
Let us provide more details on the second statement. $\cC^*$ is completely positive if and only if it can be expressed with some (unnormalized) Kraus operators $C_j$.  These transform into the Kraus operators $T_j = D^{1/2} C_j D^{-1/2}$ of $\cT^*$.

We prove that the Choi-Jamio{\l}kowski state $J(\cC^*)$ is positive semidefinite, which means that the linear map $\cC^*$ is completely positive. For each $x$, let $|p_x\>$ denote the column vector of the transition matrix $P$ corresponding to the state $x$. We can write $J(\cC^*)$ as follows:
\begin{align}
        J(\cC^*)
        &=
        \sum_{x,y} \cC^*(|x\>\<y|) \otimes |x\>\<y| \\
        &= 
        \sum_{x} \sum_{x'} 
        p_{x',x} |x'\>\<x'| \otimes |x\>\<x|
        +
        \sum_{x,y : x\neq y}
        |p_x\>\<p_y| \otimes |x\>\<y| \\
        &=
        \sum_{x,y}
        |p_x\>\<p_y| \otimes |x\>\<y| 
        +
        \sum_x \big(\mathrm{diag}(|p_x\>) - |p_x\>\<p_x|\big) \otimes |x\>\<x|.
    \end{align}
    It suffices to show that the matrices $\mathrm{diag}(|p_x\>) - |p_x\>\<p_x|$ are positive semidefinite for all $x$ since the first sum defines a positive semidefinite matrix.  This is the case because these matrices are diagonally dominant.
    
    Consider an arbitrary $x$ and the corresponding matrix $\mathrm{diag}(|p_x\>) - |p_x\>\<p_x|$. Pick any column $x'$. Its diagonal entry is
    \begin{align}
        p_{x',x} - p_{x',x}^2 &= p_{x',x} ( 1 - p_{x',x})
    \end{align}
    and its off-diagonal entries are
    \begin{align}
        p_{x',x} \cdot p_{x'',x} 
    \end{align}
    where $x''\neq x'$. We now see that the matrix is diagonally column dominant because the off-diagonal entries (which are all nonnegative) in each row sum up to the corresponding diagonal entry. 
\end{proof}

\begin{lemma}[Quantized grand coupling is completely positive]
The map $\cT$ corresponding to a grand coupling is completely positive.
\end{lemma}

\begin{proof}
We apply the matricization map (the inverse of the vectorization map) to the transition matrix of the grand coupling in (\ref{eq:C_grand}) and obtain
\begin{align}
    \cC^*(\bullet) 
    &=
    \sum_{r\in \cR} \sum_{x,y} \Pr(r) \,
    |f(x,r)\>\<x| \bullet |y\>\<f(y,r)|.
\end{align}
After applying the similarity transformation, we obtain
\begin{align}
    & \nonumber
    \cT^*(\bullet) \\ 
    &=
    \sum_{r\in \cR} \sum_{x,y} \Pr(r) \,
    D^{-1/2}|f(x,r)\>\<x|D^{1/2} \bullet D^{1/2}|y\>\<f(y,r)|D^{-1/2}.    
\end{align}
Thus, the quantized grand coupling $\cT$ can be written as 
\begin{align}
    \cT(\bullet) 
    &= 
    \sum_{r\in\cR} T_r \bullet T_r^\dagger
\end{align}
with Kraus operators
\begin{align}\label{eq:T_r}
    T_r 
    &= 
    \sqrt{\Pr(r)} \sum_{x\in\Omega} D^{1/2} |x\>\<f(x,r)| D^{-1/2}.
\end{align}
The existence of the Kraus representation implies that $\cT$ is completely positive. 
\end{proof}

It is beneficial to explicitly verify that the quantized grand coupling has $|\sqrt{\pi}\>$ as a fixed point.  Recall that this already follows from Lemma~\ref{lem:fixed_point}. We have
\begin{align}
    & \nonumber
    \cT(|\sqrt{\pi}\>\<\sqrt{\pi}|) \\
    &=
    \sum_{r\in \cR} 
    \Pr(r) \sum_{x,y\in\Omega} D^{1/2} |x\>\<f(x,r)| D^{-1/2} |\sqrt{\pi}\>
    \<\sqrt{\pi}| D^{-1/2} |f(y,r)\>\<y| D^{1/2} \\
    &=
    \sum_{r\in \cR} 
    \Pr(r) \sum_{x,y\in\Omega} D^{1/2} |x\>\<f(x,r)|e\>
    \<e|f(y,r)\>\<y| D^{1/2} \\
    &=
    \sum_{x,y\in\Omega} D^{1/2} |x\>\<y| D^{1/2} \\
    &=
    |\sqrt{\pi}\>\<\sqrt{\pi}|,
\end{align}
where $|e\>=\sum_{x\in\Omega} |x\> = (1,1,\ldots,1)^T$ denotes the all-ones column vector.

\section{Convergence rate of quantized Markov chain couplings}\label{sec:convergence}

We establish here that the channel ${\cT}$ converges to the pure fixed point $\ket{\sqrt{\pi}}$ and provide a bound on the rate of convergence from an arbitrary initial state in the trace norm distance $\|\cdot\|_1 $. Conventional approaches to bounding mixing times \cite{temme2010chi,
temme2017thermalization,bardet2022hypercontractivity,bardet2021modified,
bardet2022hypercontractivity,
bardet2024entropy} for quantum Markov chains typically rely on the assumption that the fixed point is a mixed state of full rank. When the fixed points are pure states, we need to analyze the absorption probability into invariant subspaces \cite{carbone2021absorption,frigerio1978stationary}. We will show that the correct projector $Q$ onto the subspace of invariant states under $\cT$ and its complement $Q^\perp$ are
given by
\begin{align}
    Q 
    &= 
    \proj{\sqrt{\pi}} \\
    Q^\perp 
    &= 
    \mathbb{I} - Q,
\end{align}
respectively. To identify the absorbing subspace \cite{carbone2021absorption} and prove convergence to this subspace, cf.~Lemma 3.1 and Theorem 3.2 in \cite{frigerio1978stationary}, we need to prove the following properties. 

\begin{lemma}[Fixed point projector]
Let $\cT$ be a completely positive quantized Markov chain coupling of the Markov chain $P$ with fixed point $\pi$ and coupling $C$. Then, $Q = \proj{\sqrt{\pi}}$, is the maximal projector onto the absorbing subspace that satisfies the following properties:

\begin{enumerate}
    \item The projection $Q$ is said to {\bf reduce} $\cT$ if it is globally invariant under $\cT^m$, which means 
    \begin{align}
        Q\left [\cT^* \right ]^m(A)Q = Q\left [\cT^* \right ]^m(QAQ)Q \; \; \; \;  \mbox{for all} \; \; \; \; A \in \mathbb{M}_{|\Omega| \times |\Omega|}(\mathbb{C}).
    \end{align}

    \item The projection $Q$ is  {\bf expanding} under the dynamics of $\cT$, 
    which means
        \begin{align}
        \lim_{m\rightarrow \infty} \left [\cT^* \right ]^m(Q) = \mathbb{I}.
        \end{align}    
\end{enumerate}
\end{lemma}

\begin{proof}
For (1), note that equality holds if for all states $\rho$ we have 
\begin{align}
    \tr\left[\rho  Q\left [\cT^* \right ]^m(A)Q \right] = \tr\left[\rho Q\left [\cT^* \right ]^m(QAQ) Q\right]. 
\end{align}
We can work in the dual picture now and write using the definition of $Q$ that
\begin{align}
    \tr\left[\cT^m( Q \rho  Q) A \right] =  \bra{\sqrt{\pi}} \rho \ket{\sqrt{\pi}} \tr\left[\cT^m(\proj{\sqrt{\pi}}) A \right].
\end{align}
Recall that, cf.~Lemma~\ref{lem:fixed_point}, we have that $\cT(\proj{\sqrt{\pi}}) = \proj{\sqrt{\pi}}$, so that 
\begin{align}
    \tr\left[\cT^m( Q \rho  Q) A \right] =  \bra{\sqrt{\pi}} \rho \ket{\sqrt{\pi}} \bra{\sqrt{\pi}} A \ket{\sqrt{\pi}}.
\end{align}
The same set of arguments gives rise to 
\begin{align}
    \tr\left[\rho Q\left [\cT^* \right ]^m(QAQ) Q\right] =  \bra{\sqrt{\pi}} \rho \ket{\sqrt{\pi}} \bra{\sqrt{\pi}} A \ket{\sqrt{\pi}},
\end{align}
showing that both are equal. For (2), we have
\begin{align}
    (\cT^*)^m (Q)
    &= 
    D^{-1/2} (\cC^*)^m (D^{1/2} Q D^{1/2}) D^{-1/2} \\
    &=
    D^{-1/2} (\cC^*)^m (D^{1/2} |\sqrt{\pi}\>\<\sqrt{\pi}| D^{1/2}) D^{-1/2} \\
    &=
    D^{-1/2} (\cC^*)^m (|\pi\>\<\pi|) D^{-1/2}.
\end{align}
We now use that
\begin{align}
    \lim_{m\rightarrow\infty}(\cC^*)^m
    (|\pi\>\<\pi|) 
    &=
    \sum_{x\in\Omega} \pi_x |x\>\<x| = D.
\end{align}

To derive this limes, observe that $(\cC^*)^m(|\pi\>\<\pi|)$ is equal to
\begin{align}
    C^m(|\pi\> \otimes |\pi\>)
\end{align}
in the vectorized representation. The product probability distribution $|\pi\> \otimes |\pi\>$ has $|\pi\>$ as marginal in both components. The transition matrix $C$ preserves $|\pi\>$ as marginal due to the first coupling condition and $P|\pi\>=|\pi\>$. But $C$ also coalesces the components as $m$ increases meaning that the probability of having unequal states in both components goes to $0$. The unique probability distribution that has $|\pi\>$ as marginal and does not contain unequal terms coincides with $\sum_{x} \pi_x |x\> \otimes |x\>$, which is the vectorization of $D$.

Combining the above, we obtain 
\begin{align}
    \lim_{m\rightarrow\infty}(\cT^*)^m (Q) &= \mathbb{I}.
\end{align}
This is equivalent to 
\begin{align}\label{eq:gen_Q_annihilated}
    \lim_{m\rightarrow\infty}(\cT^*)^m (Q^\perp) &=
    0,
\end{align}
which implies that $Q$ is the fixed point of $\cT$. 
\end{proof}

Having established these properties of $Q$ and $\cT$, we can now rely on the results of \cite{frigerio1978stationary} that prove that $\cT^m$ converges to the invariant subspace with the projector $Q$. In particular, these results imply that $\ket{\sqrt{\pi}}$ is the unique fixed point of the quantized Markov chain coupling $\cT$. 
However, rather than referring to \cite{frigerio1978stationary}, we take a more direct approach and follow \cite{dengis2014optimal} to derive a quantitative bound on the convergence rate, which also immediately implies that $Q$ is the correct absorbing projector for $\cT$.

\begin{lemma}[Convergence bounds]\label{lem:convergence_bound}
Let $\cT$ and $Q$ be given as above. If we can find an integer $m$ for some $\varepsilon > 0$ so that $\tr \left [Q^\perp \cT^m(\rho_0)\right ] < \varepsilon$ for all density matrices $\rho_0$, then 
\begin{align}
  \left \| \cT^k(\rho_0) - \proj{\sqrt{\pi} } \right \|_1 \leq 2 \sqrt{\varepsilon},
\end{align}
for all $k \geq m$.
\end{lemma}

\begin{proof}
We set  $\rho_m=\cT^m(\rho)$ and write $\tr[Q^\perp \rho_m] = 1 - \tr[Q \rho_m] < \varepsilon$. Now, following the argument in \cite{dengis2014optimal} we apply the gentle measurement lemma \cite{wilde2013quantum} for the projective measurement $Q$ so that we have the trace norm bound 
\begin{align}
    \left \| \rho_m - \frac{Q\rho_m Q}{\tr[Q\rho_m]} \right \|_1 \leq 2\sqrt{\varepsilon}.
\end{align}
Given our definition of $Q$, we immediately have that $(Q\rho_mQ)(\tr[Q\rho_m])^{-1} = \proj{\sqrt{\pi}}$, which gives us the bound in the lemma for all $k \geq m$ by applying the data processing inequality \cite{wilde2013quantum} and observing that the projector reduces $\cT$. \end{proof}

Using these results, we therefore need to show that  $\tr \left [Q^\perp \cT^m(\rho_0)\right ] < \varepsilon$ for pairs of $m$ and $\varepsilon$. We will establish a connection between the quantized Markov coupling $\cT$ and the coalescence time of the classical Markov chain coupling $C$. We need to introduce some terminology. We define the edge states for $x,y \in \Omega$ as
\begin{align}
    |-_{xy}\> 
    &= 
    \frac{1}{\sqrt{2}} \big(
    |x\> - |y\> \big).
\end{align}
We refer to the matrices
\begin{align}
    |-_{xy}\>\<-_{xy}|
\end{align}
as elementary Laplacians and call any linear combination of them a graph Laplacian. Note that $Q^\perp$ is related to the $\pi$-weighted graph Laplacian of the complete graph on $\Omega$. 

\begin{proposition}\label{lem:rescaled_Qperp}
The rescaled projector $D^{1/2} Q^\perp D^{1/2}$ can be expressed as convex combination of elementary Laplacians. More precisely, we have
\begin{align}
    D^{1/2} Q^\perp D^{1/2} 
    &=
    \sum_{x,y} \pi_x \pi_y |-_{xy}\>\<-_{xy}|.
\end{align}
\end{proposition}


\begin{proof}
The lemma statement is established by the following elementary calculations:
\begin{align}
    D^{1/2} Q^\perp D^{1/2} 
    &=
    D^{1/2} \big(I - |\sqrt{\pi}\>\<\sqrt{\pi}| \big) D^{1/2} \\
    &=
    D - |\pi\>\<\pi| \\
    &=
    \sum_x \pi_x |x\>\<x| - \sum_{x,y} \pi_x \pi_y |x\>\<y| \\
    &=
    \sum_{x,y} \pi_x \pi_y |x\>\<x| - \sum_{x,y} \pi_x \pi_y |x\>\<y| \\
    &=
    \frac{1}{2}
    \sum_{x,y} \pi_x \pi_y \big(|x\>\<x| - |x\>\<y| - |y\>\<x| + |y\>\<y| \big) \\
    &=
    \sum_{x,y} \pi_x \pi_y |-_{xy}\>\<-_{xy}|.
\end{align}
\end{proof}

We introduced these decompositions into graph Laplacians because the quantum map $\cC^*$ preserves this structure as shown in the lemma below.

\begin{lemma}[Preservation of graph Laplacians]
The map $\cC^*$ preserves graph Laplacians. For $x\neq y$, we have
\begin{align}
    \cC^*(|-_{x,y}\>\<-_{x,y}|)
    &=
    \sum_{x',y'} c_{(x',y'),(x,y)} 
    |-_{x',y'}\>\<-_{x',y'}|.
\end{align}
\end{lemma}

\begin{proof}
We have
\begin{align}
    |-_{x,y}\>\<-_{x,y}|
    &=
    \frac{1}{2} \big( |x\>\<x| - |x\>\<y| - |y\>\<x| + |y\>\<y| \big).
\end{align}
Let us derive how $\cC^*$ acts on these terms. We will make use of the coupling conditions stated in Section~\ref{sec:Markov_chain_couplings}.

We have
\begin{align}
    \cC^*(|x\>\<x|)
    &=
    \sum_{x'} c_{(x',x'),(x,x)}
    |x'\>\<x'| \\
    &=
    \sum_{x'} p_{x',x} |x'\>\<x'| \\
    &=
    \sum_{x'} \sum_{y'} c_{(x',y'),(x,y)} |x'\>\<x'|,
\end{align}
which holds for any $y\neq x$. To obtain this identity, we used the first coupling condition and the second coupling condition.

Analogously, we have
\begin{align}
    \cC^*(|y\>\<y|)
    &=
    \sum_{x'} \sum_{y'} c_{(x',y'),(x,y)} |y'\>\<y'|,
\end{align}
which holds for any $x\neq y$.

For $x\neq y$, we have
\begin{align}
    \cC^*(|x\>\<y|)
    &=
    \sum_{x'} \sum_{y'} c_{(x',y'),(x,y)} |x'\>\<y'| \\
    \cC^*(|y\>\<x|)
    &=
    \sum_{x'} \sum_{y'} c_{(y',x'),(y,x)} |y'\>\<x'| \\
    &=
    \sum_{x'} \sum_{y'} c_{(x',y'),(x,y)} |y'\>\<x'|
\end{align}
We used the third coupling condition to establish the last inequality.

Putting everything together using linearity, we arrive at the lemma statement.
\end{proof}

\begin{lemma}\label{lem:tr_couple}
The probability $P_{x,y}\{\tau_\mathrm{coal} > m\}$ that the two components have not coalesced after making $m$ transitions according to the transition matrix $C$ when starting in the initial state $(x,y)$ can be written as
\begin{align}
    \Pr_{x,y}\{ \tau_\mathrm{coal} > m \}
    &=
    \tr\Big(
    \big[ \cC^* \big]^m \big( |-_{x,y}\>\<-_{x,y}| \big) 
    \Big). 
\end{align}
\end{lemma}

\begin{proof}
We apply the previous lemma to evaluate $[\cC^*]^m(|-_{x,y}\>\<-_{x,y}|)$. After taking the trace, the resulting expression simplifies further because the trace of each edge Laplacian (with $x'\neq y'$) is equal to $1$. This yields
\begin{align}
    \tr\Big(
    \big[ \cC^*\big]^m \big( |-_{x,y}\>\<-_{x,y}| \big) 
    \Big) 
    &=
    \sum_{x',y'} [C^m]_{(x',y'),(x,y)} 
    \tr \big( |-_{x',y'}\>\<-_{x',y'}| \big) \\
    &=
    \sum_{x' \neq y'} [C^m]_{(x',y'),(x,y)} \\
    &=
    \Pr_{x,y}\{ \tau_\mathrm{coal} > m \},
\end{align}
where we used the simplification introduced in Remark~\ref{rem:helpful} for the final step.
\end{proof}


\begin{lemma}[Classical coalescence implies rapid quantum convergence]\label{lem:class_q_coal}
Consider an ergodic Markov chain with transition matrix $P=(p_{x',x})$ and stationary distribution $\pi = (\pi_x)$.
Assume that we are given a Markov chain coupling with transition matrix $C$ such that the corresponding quantized coupling $\cT$ is completely positive. 

Then, we have
\begin{align}
    \tr\Big( Q^\perp \cT^m(\rho_0)\Big) 
    &\le
    \frac{1}{\pi_*} \Pr_\mathrm{max} \{ \tau_\mathrm{coal} > m \},
\end{align}
where $\pi_* = \min \{\pi_x : x \in \Omega \}$.
\end{lemma}

\begin{proof}
We have
\begin{align}
    &
    \tr\Big( Q^\perp \cT^m(\rho_0)\Big) \\
    &=
    \tr\Big( \big(\cT^*\big)^m\big(Q^\perp\big) \rho_0\Big) \\
    &=
    \tr\Big( D^{-1/2}
    \big(\cC^*\big)^m\big(D^{1/2} Q^\perp D^{1/2} \big) D^{-1/2} \rho_0\Big) \\
    &\le
    \frac{1}{\pi_*}
    \tr\Big( 
    \big(\cC^*\big)^m\big(D^{1/2} Q^\perp D^{1/2} \big)
    \Big) \\
    &=
    \frac{1}{\pi_*}
    \sum_{x,y} \pi_x \pi_y \, \tr\left(
    \big( \cC^*\big)^m \big( |-_{x,y}\>\<-_{x,y}| \big) 
    \right) \\
    &=
    \frac{1}{\pi_*}
    \sum_{x,y} \pi_x \pi_y \, 
    \Pr_{x,y}\{\tau_\mathrm{coal} > m\} \\
    &\le
    \frac{1}{\pi_*} \Pr_\mathrm{max} \{\tau_\mathrm{coal} > m\}
\end{align}
First, we switch from the Schr\"odinger picture $\cT$ to the Heisenberg picture $\cT^*$. 
Second, we appeal to the fact that $\cT^*$ is obtained from $\cC^*$ by the similarity transformation $\cT^*(\bullet)=D^{-1/2} \cC^*(D^{1/2}\bullet D^{1/2})D^{-1/2}$, which facilitates the computation of the power $(\cT^*)^m$. 
Third, we rely on the bound $\tr(A B)\le \lambda_\mathrm{max}(B) \tr(A)$ that holds for arbitrary positive semidefinite matrices $A$ and $B$. We apply it with $A=(\cC^*)^m(|-_{x,y}\>\<-_{x,y}|)$ and $B=D^{-1/2}\rho_0 D^{-1/2}$ and use $\lambda_\mathrm{max}(B)\le \frac{1}{\pi_*}$.  
Fourth, we invoke Lemma~\ref{lem:rescaled_Qperp} to write $D^{1/2} Q^\perp D^{1/2}$ as a convex combination of elementary Laplacians.
Fifth, we invoke Lemma~\ref{lem:tr_couple} to write the traces of $[\cC^*]^m(|-_{x,y}\>\<-_{x,y}|)$ as $\Pr_{x,y}\{\tau_\mathrm{coal} > m\}$.
Finally, we use that the convex combination of $\Pr_{x,y}\{\tau_\mathrm{coal} > m\}$ is upper bounded by $\Pr_\mathrm{max}\{\tau_\mathrm{coal} > m\}$.
\end{proof}


\begin{lemma} \label{lem:exp_contraction}
Consider an arbitrary Markov chain with transition matrix $P$.  For any Markov chain coupling with transition matrix $C$, we have 
\begin{align}
    \Pr_\mathrm{max}\{\tau_\mathrm{coal} \ge \ell \cdot m\} 
    &\le
    \left( \Pr_\mathrm{max}\{\tau_\mathrm{coal} \ge m\} \right)^\ell.
\end{align}
\end{lemma}

\begin{proof}
The coupling transition matrix $C$ acts on state space $\Omega\times\Omega$, which we decompose into the disjoint union of $\Omega_\mathrm{diag} = \{ (x,x) : x \in \Omega \}$ and its complement denoted by $\bar{\Omega}_\mathrm{diag}$.

The matrix $C$ has the following block structure
\begin{align}
    \left(
    \begin{array}{cc}
    P & A \\
    0 & B
    \end{array}
    \right)
\end{align}
with respect to this decomposition. Observe that $P\in\R^{N\times N}$ is the transition matrix of the Markov chain. The matrices $A\in\R^{N\times \bar{N}}$ and $B\in\R^{\bar{N}\times\bar{N}}$ are some matrices. The reason for this particular block structure is that once the components meet, they stay together forever. In addition, $C$ restricted to $\Omega_\mathrm{diag}$ acts as $P$. 

The power $C^m$ has the block structure
\begin{align}
    \left(
    \begin{array}{cc}
    P^m & A^{(m)} \\
    0   & B^m
    \end{array}
    \right).
\end{align}
Observe that 
\begin{align}
    \| B^m \|_1 
    &=
    \max_{x\neq y} \sum_{x'\neq y'} [C^m]_{(x',y'),(x,y)} \\
    &=
    \max_{x\neq y} \Pr_{x,y}\{\tau_\mathrm{coal} > m\} \\
    &= 
    \Pr_\mathrm{max}\{\tau_\mathrm{coal} > m\}.
\end{align}
We assume that $m$ is chosen to be large enough so that 
\begin{align}
    \Pr_\mathrm{max}\{\tau_\mathrm{coal} > m\} \le \frac{1}{2}.
\end{align}  
The power $C^{\ell m}$ has the block structure
\begin{align}
    C^{\ell m}
    &=
    \left(
    \begin{array}{cc}
    P^{\ell m} & A^{(\ell m)} \\
    0          & B^{\ell m}
    \end{array}
    \right).
\end{align}
We have
\begin{align}
    \Pr_\mathrm{max}\{\tau_\mathrm{coal} > \ell m\} 
    &= 
    \|B^{\ell m} \|_1 \\
    &\le 
    \|B^m \|_1^\ell \\
    &=
    \left( \Pr_\mathrm{max}\{\tau_\mathrm{coal} > m\} \right)^\ell
\end{align}
since the matrix norm $\| \bullet \|_1$ is sub-multiplicative.
\end{proof}

\medskip

We can now collect the previously established results and prove the main convergence result.

\begin{theorem}[main]
Let $\cT$ denote a quantized coupling of the Markov chain $P$ with stationary distribution $\pi=(\pi_x : x \in \Omega)$. Given any initial state $\rho$, this map converges to the qsample $|\sqrt{\pi}\>\<\sqrt{\pi}|$ in trace norm distance as  

\begin{align}
   \frac{1}{2}   \big\| \cT^{m}(\rho) - |\sqrt{\pi}\>\<\sqrt{\pi}| \big\|_\mathrm{tr}
    &\le \sqrt{\varepsilon},
\end{align}
after
\begin{align}
    m \geq \frac{1}{2} \log_2 \left (\frac{1}{\varepsilon \; \pi_*}\right) \; t_\mathrm{couple}
\end{align}
steps, where $\pi_* = \min\{ \pi_x : x \in \Omega\}$. 
\end{theorem}

\begin{proof}
The result follows from a combination of the results above, when using the reducing and expanding projector $Q = \proj{\sqrt{\pi}}$ for which we know that due to Lemma~\ref{lem:class_q_coal} that 
\begin{align}
    \tr\Big( Q^\perp \cT^m(\rho_0)\Big) 
    &\le
    \frac{1}{\pi_*} \Pr_\mathrm{max} \{ \tau_\mathrm{coal} > m \}.
\end{align}
Setting m = $\ell \, t_\mathrm{couple}$, we can bound $\Pr_\mathrm{max} \{ \tau_\mathrm{coal} > m \} \leq 4^{-\ell}$ due to Lemma \ref{lem:exp_contraction} and the definition of $t_\mathrm{couple}$. Setting $\ell \geq 2^{-1} \log_2 \left (\frac{1}{\varepsilon \; \pi_*}\right)$, we have that $\tr( Q^\perp \cT^m(\rho_0)) \leq \varepsilon$ so that we can apply Lemma ~\ref{lem:convergence_bound} and read off the bound as stated in the theorem after $m$ steps. 
\end{proof}

\section{Implementation of quantized grand couplings}\label{sec:implementation}

We now show how to efficiently implement quantized grand couplings. The key insight is that the Kraus operators $T_r$ in (\ref{eq:T_r}) are sparse for many interesting Markov chains studied in the literature.

\subsection{Simulation of sparse tcp-maps}

For the sake of completeness, we briefly describe how to simulate tcp-maps given access to block-encodings of their Kraus operators. To keep the presentation simple, we assume that the block-encodings perfectly encode the Kraus operators. We refer the reader to \cite[Lemma 2]{li2023simulating} for the most general solution that examines how errors in the block-encodings affect the accuracy with which the tcp-map can be approximated.

\begin{lemma}[Simulation of tcp-maps]
Let
\begin{align}
    \cN(\bullet) 
    &= 
    \sum_{k=1}^\kappa A_k \bullet A_k^\dagger
\end{align}
be an arbitrary tcp-map acting on $\C^d$ with Kraus operators $A_k$. Assume we are given perfect block-encodings of the Kraus operators $A_k$ by the unitaries 
\begin{align}
    U_k 
    &=
    \left(
    \begin{array}{cc}
         A_k & C_k \\
         B_k & D_k
    \end{array}
    \right),
\end{align}
acting on $\C^2\otimes \C^d$ (where $B_k$, $C_k$, and $D_k$ can be any matrices acting on $\C^d$ as long as $U_k$ is unitary). Then, there exists an efficient quantum circuit that simulates the tcp-map $\cN$.
\end{lemma}

\begin{proof}
Adjoin a new control register $\C^m$. Define the controlled unitary
\begin{align}
    W &= \sum_{j=1} |j\>\<j| \otimes U_j
\end{align}
acting on $\C^\kappa \otimes \C^2 \otimes \C^d$ and the state
\begin{align}
    |\mu\> 
    &=
    \frac{1}{\sqrt{\kappa}} \sum_{k=1}^\kappa |k\>.
\end{align}
Let $|\xi\>\in\C^d$ be an arbitrary initial state. Applying $W$ to the state
\begin{align}
    |\Psi\>
    &=
    |\mu\> \otimes |0\> \otimes |\xi\>.
\end{align}
results in the state
\begin{align}
    |\Phi\>
    &=
    W |\Psi\> \\
    &=
    \frac{1}{\sqrt{\kappa}} \sum_{k=1}^\kappa |k\> \otimes U_k \big( |0\> \otimes |\xi\> \big) \\
    &=
    \frac{1}{\sqrt{\kappa}} \sum_{k=1}^\kappa |k\> \otimes |0\> \otimes A_k |\xi\> +
    \frac{1}{\sqrt{\kappa}} \sum_{k=1}^\kappa |k\> \otimes |1\> \otimes B_k |\xi\> 
\end{align}
It will be helpful to exchange the positions of the $\C^\kappa$ register and the $\C^2$ register.  After this rearrangement, we can express everything as
\begin{align}
    |\Psi\> 
    &= 
    |0\> \otimes |\psi\> \\
    |\Phi\> 
    &= 
    |0\> \otimes \frac{1}{\sqrt{\kappa}} \, |\phi_0\> +
    |1\> \otimes \sqrt{1-\frac{1}{\kappa}} \, |\phi_1\>, \label{eq:Phi}
\end{align}
where
\begin{align}
    |\psi\> 
    &= 
    |\mu\> \otimes |\xi\> \\
    |\phi_0\> 
    &= 
    \sum_{k=1}^\kappa |k\> \otimes A_k |\psi\> \\
    |\phi_1\> 
    &= 
    \frac{1}{\sqrt{\kappa-1}} 
    \sum_{k=1}^\kappa |k\> \otimes B_k |\psi\>.
\end{align}
Observe that $|\phi_0\>$ and $|\phi_1\>$ have the norm $1$.  For $|\phi_0\>$, we have
\begin{align}
    \<\phi_0|\phi_0\> 
    &= 
    \sum_{k=1}^\kappa \<\psi|A_k^\dagger A_k |\psi\> = 1
\end{align}
since the Kraus operators satisfy the condition $\sum_k A_k^\dagger A_k = I_d$.  For $\phi_1$, we have
\begin{align}
    \<\phi_1|\phi_1\> 
    &= 
    \frac{1}{\kappa-1}\sum_{k=1}^\kappa \<\psi|B_k^\dagger B_k|\psi\> = 1
\end{align}
since $\sum_{k=1}^\kappa B_k^\dagger B_k = (\kappa-1) I_d$. The latter is seen as follows. For each $\kappa$, the top-left block of $U_k^\dagger U_k$ is given by the formula $A_k^\dagger A_k + B_k^\dagger B_k$, which must be equal to $I_d$ since $U_k$ is unitary. Summing over $k=1,\ldots,\kappa$ yields $\kappa \, I_d = \sum_k A_k^\dagger A_k + B_k^\dagger B_k = I_d + \sum_k B_k^\dagger B_k$, where we used the Kraus condition.

Let us emphasize that the amplitudes $1/\sqrt{\kappa}$ and $\sqrt{1 - 1/\kappa}$ of $|\phi_0\>$ and $|\phi_1\>$, respectively, in the decomposition of $|\Phi\>$ in  (\ref{eq:Phi}) are constants (that is, they are completely independent of $|\xi\>$ even though $|\Phi\>$ depends on $|\xi\>$). This makes it possible to apply oblivious amplitude amplification to map the initial state $|\Psi\>$ to $|\Phi\>$.  The adjective ``oblivious'' in this context means that reflecting around the initial state $|\Psi\>$ is not necessary. Define the reflection $R=2 P - I$, where $P = |0\>\<0| \otimes I_\kappa \otimes I_d$, and the Grover iteration $G=-W R W^\dagger R$.  We refer the reader to \cite[Lemma 3.6]{berry14exponential}, which shows that $O(\sqrt{\kappa})$ Grover iterations suffice to map $|\Psi\>$ to $|\Phi\>$.

To summarize, the discussion above shows that we can prepare the state
\begin{align}
    |0\> \otimes \sum_{k=1}^\kappa |k\> \otimes A_k |\xi\>
\end{align}
on $\C^2 \otimes \C^\kappa \otimes \C^d$ starting from an arbitrary state $|\xi\>$ in the third register. 

It now suffices to trace out the $\C^2$ and $\C^\kappa$ registers to implement the tcp-map $\cN$. Discarding the two register prepares the state
\begin{align}
    \sum_{k=1}^\kappa A_k |\xi\>\<\xi| A_k^\dagger.
\end{align}
It is important that this works for any initial state $|\xi\>$ on $\C^d$.  
\end{proof}

In the remainder of this section, we present three examples of grand couplings from the book \cite{levin2017mixing}. The resulting Kraus operators $T_r$ are all row- and column-sparse.  We can now leverage the results in \cite{camps2024explicit, gilyen2019quantum} to efficiently block-encode such matrices.

\subsection{Random walk on the hypercube}

This is based on \cite[Example 5.3.1]{levin2017mixing}.
Consider a lazy walk on the hypercube $\{0,1\}^n$: at each step the walker remains at the current position with probability $1/2$ and with probability $1/2$ moves to a position chosen uniformly at random among the neighboring vertices (bit-strings at Hamming distance $1$).

A convenient way to generate the lazy walk is as follows: pick one of the $n$ coordinates uniformly at random, and refresh the bit at this coordinate with a random bit (one which equals $0$ or $1$ with probability $1/2$).

This method for running the walk leads to the following coupling of two walks: first, 
pick among the $n$ coordinates uniformly at random; suppose that the coordinate $i$ is selected. In both walks, replace the bit at coordinate $i$ with the same random fair bit. 

A moment's thought reveals that this is a coupling that satisfies all the conditions required to define a valid tcp-map.
For $i\in\{1,\ldots,n\}$ and $b\in\{0,1\}$ define the Kraus operator $T_{i,b}$ acting on the computational basis of $(\C^2)^{\otimes n}$ as follows
\begin{align}
T_{i,b}|x_1, \ldots, x_i, \ldots x_n\>
&=
|x_1, \ldots, b, \ldots x_n\>.
\end{align}
We now have
\begin{align}
    \cT(\bullet)
    &=
    \frac{1}{2n}
    \sum_{i=1}^n \sum_{b\in\{0,1\}}
    T_{i, b} \bullet T_{i, b}^\dagger.
\end{align}
We do not have to apply the similarity transformation because the limiting distribution is the uniform distribution. We can efficiently implement the tcp-map because the Kraus operators are $n$-sparse. 

If $\tau$ is the first time when all of the coordinates have been selected at least once, then the two walkers agree with each other from time $\tau$ onwards. The distribution of $\tau$ is exactly the same as the coupon's random variable. One can show that
\begin{align}
    \Pr\{\tau > n\log n + cn\} \le e^{-c}.
\end{align}

\subsection{Random colorings}

In this and the next subsection, we consider two chains whose state spaces are contained in a set of a form $S^V$, where $V$ is the vertex set of a graph and $S$ is a finite set. The elements of $S^V$called configurations, are the functions from $V$ to $S$. We visualize a configuration as a labeling of vertices with elements of $S$.

A proper $q$-coloring of a graph $G=(V,E)$ with vertex set $V$ and edge set $E$ is an element $x$ of $\{1,2,\ldots,q\}^V$, the set of functions from $V$ to $\{1,2,\ldots,q\}$, such that $x(v)\neq x(w)$ for all edges $\{v,w\}\in E$. 

For a given configuration $x$ and a vertex $v$, call a color $k$ allowable at $v$ if $k$ is different from all colors assigned to neighbors of $v$. That is, a color is allowable at $v$ if it does not belong to the set $\{x(w) : w \sim v\}$.

Let $\cX$ denote the collection of all proper $q$-colorings of $G$ and $\pi$ denote the uniform distribution on $\cX$.
In \cite[Example 3.5]{levin2017mixing} a \textit{Metropolis} chain for $\pi$ is constructed. A transition for this chain is made by first selecting a vertex $v$ uniformly from $V$ and then selecting a color $k$ uniformly from $\{1,\ldots,q\}$. If color $k$ is allowable at $v$, then vertex $v$ is assigned color $k$. Otherwise, no transition is made.

Analogously to a single Metropolis chain on colorings $\cX$, the grand coupling at each move generates a single vertex and color pair $(v,k)$, uniformly at random from $V\times\{1,\ldots,q\}$ and independently of the choices made in previous time steps. For each $x\in\cX$, the coloring is updated by attempting to re-color vertex $v$ with color $k$, accepting the update if and only if the proposed new color $k$ is allowable at $v$. (If a re-coloring is not accepted, then chain remain in its current state).

For two different colorings $x,y\in\cX$, let $X_m^x$ and $X_m^y$ denote the state of the grand coupling after $m$ steps when starting the two components in $x$ and $y$. One can show that 
\begin{align}
    \Pr\{X_m^x\neq X_m^y\} 
    &\le
    n \;
    \exp\big(
    -m \;
    c_\mathrm{met}(\Delta, q) / n
    \big),
\end{align}
where $\Delta$ denotes the maximum degree in the graph and $c_\mathrm{met} = 1 - (3 \Delta) / q$.
That is, as long as $q>3 \Delta$, the probability that the grand coupling has not coalesced decays exponentially in the number of time steps $m$. 

We can now write the quantized coupling acting on $(\C^q)^{\otimes n}$ by 
\begin{align}
    \cT(\bullet) 
    &= 
    \frac{1}{qn} \sum_{v\in V} \sum_{k=1}^q T_{(v,k)} \bullet T_{(v,k)}^\dagger.
\end{align}
The (unnormalized) Kraus $T_{(v,k)}$ are defined below. 
For a coloring $x\in\{1,\ldots,q\}^V$ (that does not need to be proper), define the operator $T_{(v,k)}$ by its action on the computational basis states according to the three cases:
\begin{itemize}
\item $T_{(v,k)} |x\> = |x\>$ if $x$ is not a proper colorings
\item $T_{(v,k)} |x\> = |x\>$ if $x$ is a proper coloring, but color $k$ is not allowable at $v$
\item $T_{(v,k)} |x\> = |x'\>$ if $x$ is a proper coloring and re-coloring vertex $v$ with color $k$ yields the proper coloring $x'$
\end{itemize}
We do not have to apply a similarity transformation because the limiting distribution $\pi$ is uniform. 
The resulting tcp-map can be implemented efficiently because the Kraus operators are $(qn)$-sparse.

\subsection{Hardcore model with fugacity}

A hardcore configuration on the graph $G=(V,E)$ is a function 
$x\in \{0,1\}^V$ satisfying $x(v)x(w)=0$ whenever $\{v,w\}\in E$; that is, 
no two adjacent vertices are simultaneously occupied. A vertex $v$ with 
$x(v)=1$ is called occupied, and a vertex $v$ with $x(v)=0$ is called vacant. Let $\cX$
denote the set of all hardcore configurations on $G$.

The hardcore model with fugacity $\lambda$ is the probability distribution $\pi$ on hardcore configurations $x\in\cX$ defined by
\begin{align}
    \pi(x) &=
    \frac{\lambda^{\sum_{v\in V} x(v)}}{Z(\lambda)}.
\end{align}
The partition function $Z(\lambda)=\sum_{x\in\cX} \lambda^{\sum_{v\in V} x(v)}$
normalizes $\pi$ to have unit total mass.

In \cite[3.3.4]{levin2017mixing} a Glauber chain for $\pi$ is constructed. It updates a configuration $X_m$ to a new configuration $X_{m+1}$ as follows: 

If there is a neighbor $w'$ of $w$ such that $X_m(w')=1$, then set $X_{m+1}(w)=0$; otherwise, let
\begin{align}
    X_{m+1}(w) = \left\{
    \begin{array}{ll}
    1 & \mbox{with probability $\lambda/(1+\lambda)$,} \\
    0 & \mbox{with probability $1/(1+\lambda)$.} 
    \end{array}
    \right.
\end{align}
Furthermore, set $X_{m+1}(v)=X_m(v)$ for all $v\neq w$.

We use the grand coupling which is run as follows: a vertex $v$ is selected uniformly at random, and a coin with probability $\lambda/(1+\lambda)$ of heads is tossed, independently of the choice of $v$. Each hardcore configuration $x$ is updated using $v$ and the result of the coin toss. If the coin lands tails, any particle present at $v$ is removed. If the coin lands heads and all neighbors of $v$ are unoccupied in the configuration $x$, then a particle is placed at $v$.

The arguments in the proof of \cite[Theorem 5.9]{levin2017mixing} show that
\begin{align}
    \Pr\{X_m^x \neq X_m^y\} 
    &\le
    n \exp\big(- m \; c_H(\lambda) / n \big),
\end{align}
where $c_H(\lambda)=[1 + \lambda(1-\Delta)]/(1+\lambda)$.

Similarly to the previous example, we see that the resulting tcp-map can be realized with $n$-sparse Kraus operators. However, in the present example we need to perform the similarity transformation because the stationary distribution $\pi$ is not uniform. This can be done efficiently because it is easy to compute the ratios $\pi(x)/\pi(y)$ for two hardcore configurations. The hard-to-compute partition function $Z(\lambda)$ cancels out, and we are left with 
\begin{align}
    \lambda^{\sum_{v\in V} x(v) - y(v)}.
\end{align}

\subsection{General theorem}

It is now straightforward to see that the following general result holds.

\begin{theorem}[Efficient implementation of quantized grand couplings]
Assume that we have a random mapping representation such that
\begin{itemize}
    \item we can efficiently compute the image and pre-images of the function $f : \Omega \times \cR \rightarrow \Omega$
    \item we can efficiently compute the probabilities $\Pr(r)$
    \item we can efficiently compute the ratios $\pi(x)/\pi(y)$ for the limiting distribution. 
\end{itemize}
Then, the Kraus operators
\begin{align}
    T_r 
    &= 
    \sqrt{\Pr(r)} \sum_{x\in\Omega} D^{1/2} |x\>\<f(x,r)| D^{-1/2}.
\end{align}
of the quantized coupling $\cT$ are $|\cR|$-sparse. Thus, $\cT$ can be implemented efficiently.
\end{theorem}

\section{Conclusions and open problems}

We have presented a novel method for efficiently preparing qsamples of stationary distributions of certain ergodic Markov chains. Our method can be turned in to an efficient quantum algorithm when there exists a grand coupling that is sparse and fast coalescing.

Our quantization serves a purpose different from Szegedy's Markov chain quantization  \cite{szegedy2004quantum}, which establishes a connection between the classical mixing time (through the spectral gap) of reversible ergodic Markov chains and the complexity of block-encoding a reflection around the qsamples (or equivalently, a projection onto the qsample). 

Let us briefly discuss how Szegedy method can be used to prepare qsamples. In general, it will be difficult to find good initial state that has sufficiently large over the qsample. Therefore, a Grover-type algorithm will take long to prepare the desired state. However, there are situations when we have a sequence of slowly varying Markov chains. Here slowly varying means that the stationary distributions of adjacent chains are close to each other in total variation distance. The latter implies that the corresponding qsamples have large overlap. This observation has been leveraged in \cite{wocjan2008speedup} to prepare qsamples.\\

We propose a few  open research problem to explore in the context of the approach presented here. 

\begin{enumerate}

    \item We have seen that not all couplings of Markov chains give rise to completely positive maps, cf.~Appendix. However, the more common ones, such as the trivial product coupling, cf. Definition~\ref{rem:indepedent_C}, and the Grand canonical coupling in Definition~\ref{rem:grand_C} do. There are considerably more general couplings that can be constructed and a very natural question to ask is, what other coupling gives rise to a tcp-map following our construction. Moreover, what couplings give rise to maps that can efficiently be implemented as a quantum circuit? 

    \item The contraction bound in Lemma ~\ref{lem:class_q_coal} reads $\tr\left(Q^\perp \cT^m(\rho_0)\right) \leq \pi_*^{-1}$ $\Pr_\mathrm{max} \{ \tau_\mathrm{coal} > m \}$. Note that since $Q^\perp$ is projector the value of $\tr\left(Q^\perp \cT^m(\rho_0)\right) \leq 1$ is actually always bounded by $1$. The prefactor $\pi_*^{-1}$, therefore seems like a substantial overestimation, that ultimately leads to an increase in the convergence time by $\log_2(\pi_*^{-1})$ when compared to the classical coupling time. The immediate question is whether this prefactor is only owed to our proof techniques and can be dropped, or whether it is indeed fundamental. 
    
    \item We have discussed the relationship and differences between our construction and the qsample reflection by using Szegedy walks. The work \cite{wocjan2023szegedy} introduced a generalization of Szegedy walk unitaries that apply to quantum detailed balanced Lindbladians, such as the Davies generator \cite{davies1976quantum}. Here, one reflects about a purification of thermal state fixed point instead of a qsample of a stationary distribution. In some sense such a purification can be seen as a generalization of qsamples. One can now ask, whether it is possible to construct a tcp-map that converges to this purification instead of the mixed state by making similar arguments as done in this work.  

\end{enumerate}

\section*{Appendix}

We now give an example of a coupling that does yield a valid tcp-map. This coupling is based on \cite[Example 5.3.1]{levin2017mixing}.

Consider a random walk on the $n$-cycle. The underlying graph of this walk has vertex set $\mathbb{Z}_n=\{0,1,\ldots,n-1\}$ and edges between $j$ and $k$ whenever $k\equiv j \pm 1$.

We consider the lazy $(p,q)$-biased walk, which remains in its current position with probability $1/2$, moves clockwise with probability $p/2$ and counterclockwise with probability $q/2$. Here $p+q=1$, and we allow the unbiased case $p=q=\frac{1}{2}$.

We construct a coupling of two particles performing lazy
walks on $\mathbb{Z}_n$. In this coupling, the two particles will never move simultaneously, ensuring that they will not jump over one another when they come to within unit distance. Until the two particles meet, at each unit of time, a fair coin is tossed, independent of all previous tosses,
to determine which of the two particles will jump. The particle that is selected makes a clockwise increment with probability $p$ and a counter-clockwise increment with probability $q$. Once the two particles collide, thereafter they make identical moves.

The formula for the Choi-Jami{\l}kowski state is 
\begin{align}
    &
    J(\cC^*_n) \\
    &=
    \sum_{x,y=0}^n |x\>\<y| \otimes \cC^*(|x\>\<y|) \\ \nonumber
    &=
    \sum_{x} |x\>\<y| \otimes \Big( 
    p |x + 1\>\<y + 1|) +
    q |x - 1\>\<y + 1|
    \Big) + \\
    &\phantom{=} 
    \sum_{x\neq y} |x\>\<y| \otimes \tfrac{1}{2} \Big( 
    p \big( |x + 1\>\<y| + |x\>\<y + 1| \big) +
    q \big( |x - 1\>\<y| + |x\>\<y - 1| \big)
    \Big)
\end{align}
It can be shown that the coupling for the unbiased case does not give rise to a valid tcp-map. This is seen by computing the Choi-Jami{\l}kowski state for a small example such as $n=3$ with the help of a simple code and noting that it has a negative eigenvalue.  The matrix is 
\begin{align}
J(\cC^*_3) 
&=
\begin{bmatrix}
0.5 & 0 & 0 & 0.5 & 0 & 0.5 & 0.5 & 0.5 & 0 \\
0 & 0.25 & 0 & 0 & 0.5 & 0 & 0 & 0 & 0.5 \\
0 & 0 & 0.25 & 0 & 0.5 & 0 & 0 & 0 & 0.5 \\
0.5 & 0 & 0 & 0.25 & 0 & 0 & 0 & 0 & 0.5 \\
0 & 0.5 & 0.5 & 0 & 0.5 & 0 & 0.5 & 0.5 & 0 \\
0.5 & 0 & 0 & 0 & 0 & 0.25 & 0 & 0 & 0.5 \\
0.5 & 0 & 0 & 0 & 0.5 & 0 & 0.25 & 0 & 0 \\
0.5 & 0 & 0 & 0 & 0.5 & 0 & 0 & 0.25 & 0 \\
0 & 0.5 & 0.5 & 0.5 & 0 & 0.5 & 0 & 0 & 0.5 \\
\end{bmatrix}
\end{align}
Its eigenvalues (rounded to $2$ digits) are given by
\begin{align}
\begin{bmatrix}
-1.0(4) & -0.3(4) & -0.3(4) & 0.2(5) & 0.2(5) & 0.2(5) & 1.0(9) & 1.0(9) & 1.7(9)
\end{bmatrix}
\end{align}
implying that the linear map $\cC^*_3$ is not completely positive.



\begin{thebibliography}{10}

\bibitem{childs2004spatial}
Andrew~M. Childs and Jeffrey Goldstone.
\newblock ``Spatial search by quantum walk''.
\newblock \href{https://dx.doi.org/10.1103/PhysRevA.70.022314}{Phys. Rev. A
  {\bf 70}, 022314}~(2004).

\bibitem{szegedy2004quantum}
Mario Szegedy.
\newblock ``Quantum speed-up of Markov chain based algorithms''.
\newblock In 45th Annual IEEE Symposium on Foundations of Computer Science.
\newblock \href{https://dx.doi.org/10.1109/FOCS.2004.53}{Pages 32--41}.
\newblock ~(2004).

\bibitem{wolf2011url}
Michael~M. Wolf.
\newblock ``Quantum channels and operations -- guided tour''.
\newblock Lecture notes~(2012).

\bibitem{verstraete2009quantum}
Frank Verstraete, Michael~M. Wolf, and J.~Ignacio Cirac.
\newblock ``Quantum computation and quantum-state engineering driven by
  dissipation''.
\newblock \href{https://dx.doi.org/10.1038/nphys1342}{Nature Physics {\bf 5},
  633--636}~(2009).

\bibitem{diehl2008quantum}
Sebastian Diehl, Andrea Micheli, Alexej Kantian, Christina Kraus, Hans~Peter
  B{\"u}chler, and Peter Zoller.
\newblock ``Quantum states and phases in driven open quantum systems with cold
  atoms''.
\newblock \href{https://dx.doi.org/10.1038/nphys1073}{Nature Physics {\bf 4},
  878--883}~(2008).

\bibitem{temme2011quantum}
Kristan Temme, Tobias~J. Osborne, Karl Gerd~H. Vollbrecht, David Poulin, and
  Frank Verstraete.
\newblock ``Quantum Metropolis sampling''.
\newblock \href{https://dx.doi.org/10.1038/nature09770}{Nature {\bf 471},
  87--90}~(2011).

\bibitem{rall2023thermal}
Patrick Rall, Chunhao Wang, and Pawel Wocjan.
\newblock ``Thermal {S}tate {P}reparation via {R}ounding {P}romises''.
\newblock \href{https://dx.doi.org/10.22331/q-2023-10-10-1132}{{Quantum} {\bf
  7}, 1132}~(2023).

\bibitem{chen2023quantum}
Chi-Fang Chen, Michael~J. Kastoryano, Fernando G. S.~L. Brandão, and András
  Gilyén.
\newblock ``Quantum thermal state preparation''~(2023).
\newblock  \href{http://arxiv.org/abs/2303.18224}{arXiv:2303.18224}.

\bibitem{jiang2024quantum}
Jiaqing Jiang and Sandy Irani.
\newblock ``Quantum Metropolis sampling via weak measurement''~(2024).
\newblock  \href{http://arxiv.org/abs/2406.16023}{arXiv:2406.16023}.

\bibitem{levin2017mixing}
David~A. Levin and Yuval Perez.
\newblock ``Markov chains and mixing times: Second edition''.
\newblock American Mathematical Society. ~(2017).
\newblock  url:~\url{https://pages.uoregon.edu/dlevin/MARKOV/markovmixing.pdf}.

\bibitem{aharonov2003adiabatic}
Dorit Aharonov and Amnon Ta-Shma.
\newblock ``Adiabatic quantum state generation''.
\newblock \href{https://dx.doi.org/10.1137/060648829}{SIAM Journal on Computing
  {\bf 37}, 47--82}~(2007).

\bibitem{propp1996exact}
James~Gary Propp and David~Bruce Wilson.
\newblock ``Exact sampling with coupled Markov chains and applications to
  statistical mechanics''.
\newblock Random Struct. Algorithms {\bf 9}, 223–252~(1996).

\bibitem{bennett1997strengths}
Charles~H. Bennett, Ethan Bernstein, Gilles Brassard, and Umesh Vazirani.
\newblock ``Strengths and weaknesses of quantum computing''.
\newblock \href{https://dx.doi.org/10.1137/S0097539796300933}{SIAM Journal on
  Computing {\bf 26}, 1510--1523}~(1997).

\bibitem{arunachalam2022simpler}
Srinivasan Arunachalam, Vojtech Havlicek, Giacomo Nannicini, Kristan Temme, and
  Pawel Wocjan.
\newblock ``Simpler (classical) and faster (quantum) algorithms for Gibbs partition functions''.
\newblock \href{https://dx.doi.org/10.22331/q-2022-09-01-789}{Quantum {\bf 6},
  789}~(2022).
\newblock  \href{http://arxiv.org/abs/2009.11270v5}{arXiv:2009.11270v5}.

\bibitem{harrow2023adaptive}
Aram~W. Harrow and Annie~Y. Wei.
\newblock ``Adaptive quantum simulated annealing for Bayesian inference and
  estimating partition functions''.
\newblock In Proceedings of the Fourteenth Annual ACM-SIAM Symposium on
  Discrete Algorithms.
\newblock \href{https://dx.doi.org/10.1137/1.9781611975994.12}{Pages 193--212}.
\newblock SIAM~(2023).

\bibitem{chakrabarti2023}
Shouvanik Chakrabarti, Andrew~M. Childs, Shih-Han Hung, Tongyang Li, Chunhao
  Wang, and Xiaodi Wu.
\newblock ``Quantum algorithm for estimating volumes of convex bodies''.
\newblock \href{https://dx.doi.org/10.1145/3588579}{ACM Transactions on Quantum
  Computing {\bf 4}, 1--60}~(2023).

\bibitem{childs2022}
Andrew~M. Childs, Tongyang Li, Jin-Peng Liu, Chunhao Wang, and Ruizhe Zhang.
\newblock ``Quantum algorithms for sampling log-concave distributions and
  estimating normalizing constants''.
\newblock In Proceedings of the 36th International Conference on Neural
  Information Processing Systems.
\newblock Pages 1--13.
\newblock NIPS '22. Curran Associates Inc.~(2022).
\newblock  url:~\url{https://papers.nips.cc/paper_files/paper/2022/hash/1686}.

\bibitem{ozgul2024}
Guneykan Ozgul, Xiantao Li, Mehrdad Mahdavi, and Chunhao Wang.
\newblock ``Stochastic quantum sampling for non-logconcave distributions and
  estimating partition functions''.
\newblock In Proceedings of the 41st International Conference on Machine
  Learning.
\newblock Volume 235 of Proceedings of Machine Learning Research, pages
  38953--38982.
\newblock PMLR~(2024).
\newblock  url:~\url{https://proceedings.mlr.press/v235/ozgul24a.html}.

\bibitem{leng2025}
Jiaqi Leng, Zhiyan Ding, Zherui Chen, and Lin Lin.
\newblock ``Operator-level quantum acceleration of non-logconcave
  sampling''~(2025).
\newblock  \href{http://arxiv.org/abs/2505.05301}{arXiv:2505.05301}.

\bibitem{wocjan2008speedup}
Pawel Wocjan and Anura Abeyesinghe.
\newblock ``Speedup via quantum sampling''.
\newblock \href{https://dx.doi.org/10.1103/PhysRevA.78.042336}{Phys. Rev. A
  {\bf 78}, 042336}~(2008).

\bibitem{temme2010chi}
Kristan Temme, Michael~J. Kastoryano, Mary~Beth Ruskai, Michael~M. Wolf, and
  Frank Verstraete.
\newblock ``The $\chi^2$-divergence and mixing times of quantum markov
  processes''.
\newblock \href{https://dx.doi.org/10.1063/1.3511335}{Journal of Mathematical
  Physics {\bf 51}, 122201}~(2010).

\bibitem{temme2017thermalization}
Kristan Temme.
\newblock ``Thermalization time bounds for Pauli stabilizer hamiltonians''.
\newblock \href{https://dx.doi.org/10.1007/s00220-016-2746-0}{Communications in
  Mathematical Physics {\bf 350}, 603--637}~(2017).

\bibitem{bardet2022hypercontractivity}
Igor Bardet and Cambyse Rouzé.
\newblock ``Hypercontractivity and logarithmic Sobolev inequality for
  non-primitive quantum Markov semigroups and estimation of decoherence
  rates''.
\newblock \href{https://dx.doi.org/10.1007/s00023-022-01196-8}{Annales Henri
  Poincaré {\bf 23}, 3839--3903}~(2022).

\bibitem{bardet2021modified}
Ivan Bardet, Ángela Capel, Angelo Lucia, David Pérez-García, and Cambyse
  Rouzé.
\newblock ``On the modified logarithmic Sobolev inequality for the heat-bath
  dynamics for 1d systems''.
\newblock \href{https://dx.doi.org/10.1063/1.5142186}{Journal of Mathematical
  Physics {\bf 62}, 061901}~(2021).

\bibitem{bardet2024entropy}
Ivan Bardet, Ángela Capel, Lin Gao, Angelo Lucia, David Pérez-García, and
  Cambyse Rouzé.
\newblock ``Entropy decay for davies semigroups of a one dimensional quantum
  lattice''.
\newblock \href{https://dx.doi.org/10.1007/s00220-023-04869-5}{Communications
  in Mathematical Physics {\bf 405}, 42}~(2024).

\bibitem{carbone2021absorption}
Raffaele Carbone and Fabio Girotti.
\newblock ``Absorption in invariant domains for semigroups of quantum
  channels''.
\newblock \href{https://dx.doi.org/10.1007/s00023-021-01016-5}{Annales Henri
  Poincaré {\bf 22}, 2497--2530}~(2021).

\bibitem{frigerio1978stationary}
Alberto Frigerio.
\newblock ``Stationary states of quantum dynamical semigroups''.
\newblock \href{https://dx.doi.org/10.1007/BF01196936}{Communications in
  Mathematical Physics {\bf 63}, 269--276}~(1978).

\bibitem{dengis2014optimal}
John Dengis, Robert König, and Fernando Pastawski.
\newblock ``An optimal dissipative encoder for the toric code''.
\newblock \href{https://dx.doi.org/10.1088/1367-2630/16/1/013023}{New Journal
  of Physics {\bf 16}, 013023}~(2014).

\bibitem{wilde2013quantum}
Mark~M. Wilde.
\newblock ``Quantum information theory''.
\newblock
  \href{https://dx.doi.org/https://doi.org/10.1017/CBO9781139525343}{Cambridge
  University Press}. ~(2013).

\bibitem{li2023simulating}
Xiantao Li and Chunhao Wang.
\newblock ``{Simulating Markovian Open Quantum Systems Using Higher-Order
  Series Expansion}''.
\newblock In Kousha Etessami, Uriel Feige, and Gabriele Puppis, editors, 50th
  International Colloquium on Automata, Languages, and Programming (ICALP
  2023).
\newblock \href{https://dx.doi.org/10.4230/LIPIcs.ICALP.2023.87}{Volume 261 of
  Leibniz International Proceedings in Informatics (LIPIcs), pages
  87:1--87:20}.
\newblock Dagstuhl, Germany~(2023). Schloss Dagstuhl -- Leibniz-Zentrum f{\"u}r
  Informatik.

\bibitem{berry14exponential}
Dominic~W. Berry, Andrew~M. Childs, Richard Cleve, Robin Kothari, and
  Rolando~D. Somma.
\newblock ``Exponential improvement in precision for simulating sparse
  hamiltonians''.
\newblock In Proceedings of the Forty-Sixth Annual ACM Symposium on Theory of
  Computing.
\newblock \href{https://dx.doi.org/10.1145/2591796.2591854}{Page 283–292}.
\newblock STOC '14New York, NY, USA~(2014). Association for Computing
  Machinery.

\bibitem{camps2024explicit}
Daan Camps, Lin Lin, Roel Van~Beeumen, and Chao Yang.
\newblock ``Explicit quantum circuits for block encodings of certain sparse
  matrices''.
\newblock \href{https://dx.doi.org/10.1137/22M1484298}{SIAM Journal on Matrix
  Analysis and Applications {\bf 45}, 801--827}~(2024).
\newblock
  \href{http://arxiv.org/abs/https://doi.org/10.1137/22M1484298}{arXiv:https://doi.org/10.1137/22M1484298}.

\bibitem{gilyen2019quantum}
András Gilyén, Yuan Su, Guang~Hao Low, and Nathan Wiebe.
\newblock ``Quantum singular value transformation and beyond: exponential
  improvements for quantum matrix arithmetics''.
\newblock In Proceedings of the 51st Annual ACM SIGACT Symposium on Theory of
  Computing.
\newblock \href{https://dx.doi.org/10.1145/3313276.3316366}{Pages 193--204}.
\newblock STOC 2019New York, NY, USA~(2019). Association for Computing
  Machinery.

\bibitem{wocjan2023szegedy}
Pawel Wocjan and Kristan Temme.
\newblock ``Szegedy walk unitaries for quantum maps''.
\newblock \href{https://dx.doi.org/10.1007/s00220-023-04797-4}{Communications
  in Mathematical Physics {\bf 402}, 3201--3231}~(2023).

\bibitem{davies1976quantum}
E.~B. Davies.
\newblock ``Quantum theory of open systems''.
\newblock Academic Press. ~(1976).

\end{thebibliography}

\end{document}